\documentclass[11pt,titlepage]{article}
\usepackage{amsthm}

\newtheorem{proposition}{Proposition}
\newtheorem{definition}{Definition}
\newtheorem{theorem}{Theorem}

\newtheorem{remark}{Remark}

\usepackage[latin1]{inputenc}
\usepackage{rotating, rotfloat}
\usepackage[round]{natbib} 
\usepackage{graphics, graphicx, hyperref, url, amsmath, subfigure, amsthm, amsfonts, enumerate, epstopdf, booktabs, bm}
\hypersetup{
    pdfstartview={FitH},    
    colorlinks=true,        
    linkcolor=blue,         
    citecolor=blue,         
    filecolor=blue,         
    urlcolor=blue           
}
\def\u{\textbf{\textit{u}}} \def\x{\textbf{\textit{x}}}  \def\z{\textbf{\textit{z}}}  \def\vv{\textbf{\textit{v}}}
\def\y{\textbf{\textit{y}}} \def\b{\textbf{\textit{b}}}  \def\0{\boldsymbol 0} \def\w{\textbf{\textit{w}}}
 \def\t{\textbf{\textit{t}}}
 
\def\etal{\itshape et al.}

\newcommand{\D}{\mbox{\bf D}}

\newcommand{\RD}{{\mbox{\bf RD}}}
\newcommand{\HD}{{\mbox{\bf HD}}}
\newcommand{\SD}{{\mbox{\bf SD}}}
\newcommand{\PD}{{\mbox{\bf PD}}}

\newcommand{\ZD}{{\mbox{\bf ZD}}}
\newcommand{\HRD}{{\mbox{\bf HRD}}}
\newcommand{\SRD}{{\mbox{\bf SRD}}}
\newcommand{\PRD}{{\mbox{\bf PRD}}}
\newcommand{\ZRD}{{\mbox{\bf ZRD}}}

\newcommand{\bec}{\begin{center}}
\newcommand{\enc}{\end{center}}
\newcommand{\bee}{\begin{eqnarray*}}
\newcommand{\ene}{\end{eqnarray*}}
\newcommand{\beq}{\begin{equation}}
\newcommand{\eeq}{\end{equation}}

\begin{document}

\title{\bf General notions of regression depth function}
\vskip 5mm

\author {{Xiaohui Liu$^{a, b}$   \footnote{Corresponding author's email: csuliuxh912@gmail.com.}
        Yuanyuan Li$^{a, b}$,
        }\\ \\[1ex]
        {\em\footnotesize $^a$ School of Statistics, Jiangxi University of Finance and Economics, Nanchang, Jiangxi 330013, China}\\
        {\em\footnotesize $^b$ Research Center of Applied Statistics, Jiangxi University of Finance and Economics, Nanchang,}\\ {\em\footnotesize Jiangxi 330013, China}\\
}

\maketitle

\begin{center}
{\sc Abstract}
\end{center}

As a measure for the centrality of a point in a set of multivariate data, statistical depth functions play important roles in multivariate analysis, because one may conveniently construct descriptive as well as inferential procedures relying on them. Many depth notions have been proposed in the literature to fit to different applications. However, most of them are mainly developed for the location setting. In this paper, we discuss the possibility of extending some of them into the regression setting. A general concept of regression depth function is also provided. 

\vspace{2mm}

\vspace{2mm}


\setlength{\baselineskip}{1.5\baselineskip}

\vskip 0.1 in
\section{Introduction}
\paragraph{}
\vskip 0.1 in \label{Introduction}

Equipping the data set with a proper ordering can bring great convenience to statistical inferences. For one dimensional data, a linear ordering exists. Hence, it is easy for practitioners to construct some descriptive as well as inferential procedures based on this ordering. However, since no natural ordering exists in spaces with dimension $d > 1$, how to order the multivariate observations is not trivial.

To this end, \cite{Tuk1975} heuristically introduced a center-outward ordering. He first defined a so-called depth function, i.e.,
\begin{eqnarray*}
  \HD(\x, P) = \inf_{\u \in \mathcal{S}^{d-1}} P(\u^\top X \le \u^\top \x) \text{ with } \mathcal{S}^{d-1} = \{\vv \in \mathbb{R}^d: \|\vv\| = 1\} \text{ and } d \ge 1,
\end{eqnarray*}
to measure the centrality of an any given point $\x \in \mathbb{R}^d$ with respect to the probability measure $P$ related to the random vector $X$. This measure of the centrality of $\x$ is actually a score for $\x$, which is one dimensional. Using this score, one then is able to order the $d$-variate observations. In the literature, this function is usually referred to as halfspace depth, but also Tukey depth in order to reflect the seminal work of Tukey. Simply replacing $P$ with its empirical version $P_n$, one may obtain the sample halfspace depth.

It is well known that halfspace depth enjoys many desirable properties. Using it, it is easy to extend many famous univariate estimators to the multivariate setting, such as univariate median, i.e.,
\begin{eqnarray}
\label{med}
    \text{Med}(\mathbf{p}_{n}) = \frac{\mathbf{z}_{(\lfloor (n + 1) / 2 \rfloor)} + \mathbf{z}_{(\lfloor (n + 2) / 2 \rfloor)}}{2}.
\end{eqnarray}
Note that it happens to be the average of all points that maximize the halfspace depth function related to the univariate random sample $\mathbf{z}^n = \{\mathbf{z}_1, \mathbf{z}_2, \cdots, \mathbf{z}_n\}$ with $\mathbf{z}_{(1)}, \mathbf{z}_{(2)}, \cdots, \mathbf{z}_{(n)}$ being the corresponding ordered statistics such that $\mathbf{z}_{(1)} \leq \mathbf{z}_{(2)} \leq \cdots \leq \mathbf{z}_{(n)}$, where $\lfloor \cdot \rfloor$ denotes the floor function, and $\mathbf{p}_{n}$ the empirical probability measure related to $\mathbf{z}^n$. Hence, it is naturally to define the average of all maximizers of $\HD(\x, P)$ as the median for a given $P$ for $d \ge 1$. Compared to the sample mean, one merit of $\text{Med}(\mathbf{p}_{n})$ is its high breakdown point robustness against to a great proportion of outliers \citep{LZW2017}.

Following the same fashion to \cite{Tuk1975}, many other depth notions have also been proposed in the literature. Among them, the most famous are simplicial depth \citep{Liu1990}, zonoid depth \citep{KM1997}, projection depth \citep{Zuo2003},  Rayleigh depth \citep{Zuo2003, HWW2011}, and so on. These depth notions differ themselves from each other in properties like robustness, continuity, and sensitivity regarding the data. For example, the sample version of halfspace depth, simplicial depth, and zonoid depth is discontinuous, while that of both projection depth and Rayleigh depth is continuous. The median-like estimators related to zonoid depth and Rayleigh depth is the sample mean, which is sensitive to few outliers. In constant, similar estimators related to the halfspace depth, simplicial depth and projection depth are well known to be robust \citep{Zuo2003}, etc. To be convenient for preferring one such function over another among different depth notions, an axiomatic definition for the general notions of statistical depth function is given in \cite{ZS2000}.

Since their introduction, statistical depth functions have proved extremely useful in various applications. To name but a few, using the depth-based ordering, one may easily weight/trim observations and in turn build many estimators alterative to the conventional sample mean or covariance, and construct some test procedures for multivariate location and scale; see for example \cite{Zuoetal2004, Mas2009, ZC2005, CSF2011} and references therein for details.

Furthermore, note that statistical depths play a similar role as quantile, which characterizes the underlying distribution \citep{KZ2010}. \cite{LPS1999} suggested a useful graphical tool, namely, $DD$-plot (depth versus depth plot), which may be utilized to compare two multivariate probability distributions by plotting their depths against each other, generalizing the well-known $QQ$-plot (quantile versus quantile plot) into multivariate descriptive statistics. \cite{YehS97} developed a depth-based bootstrap method for constructing confident regions with the critical value determined by halfspace depths of estimators based on the bootstrap samples. This kind of bootstrap regions has a data-dependent shape, and usually performs better than the elliptical regions derived from the normal distributions; see also \cite{WL2012} for similar discussions for some other depth functions. \cite{ZS2000a} proposed using the area of the depth region, which is actually a generalization of the length of the interval formed by two quantiles, to measure the dispersion of a distribution, etc.

Observe that the smaller the depth of a point is, the more outside the position of this point lies related to the data cloud. Using this fact, many classifiers have been developed based on depth functions \citep{Liu1990, GC2005, LiCAL12, LMM2014}.  Bagdistance has also been suggested by \cite{HRS2016} as an alternative to the Euclidean distance to measure the distance of a point to classes. Using this kind of distance, depth-based \emph{kNN}-methods can be developed \citep{PV2015}. Usually, these classifiers are affine-invariant, and robust if a robust depth, such as halfspace depth and projection depth, is employed.

Other applications of depth functions and their induced ordering, including risk measurement \citep{CM2007}, have also been intensively investigated in the literature. Recently, \cite{ELL2015} successfully applied extreme value statistics to refine the empirical half-space depth in ``the tail''. Extreme depth-based quantile regions are also induced, and applied to extreme analysis \citep{HE2017}. For more discussions about depth functions, we refer readers to \cite{Mos2013} for a early summary. An updated short survey about this topic can be found in \cite{RH2015}.

Since statistical depth functions have found themselves so useful in the location setting, a natural question that arises is: How to extend the aforementioned depth notions into the regression setting? Excellent works in that direction have been pioneered by \cite{RouHub1999}, which extended halfspace depth to the regression setting and introduced ``regression depth'' with a connection to the Daniels's test. This regression depth has been paid much attention since its introduction. It turns out that the median-like regression estimator induced from this depth is robust with breakdown point value converging almost surely to $1/3$ in any dimension \citep{VSHS2002}. Using this depth, one is also able to build classification methods \citep{CFJ2002}, and test procedures \citep{Wil2010}. However, the possibility of extending the other depth notions to the regression setting is not discussed yet in the literature. This motivates us to further consider the current research in this paper.

The rest of this paper is organized as follows. We first review the statistical depth functions in the location setting in Section \ref{Overview}. Then along the line of Section \ref{Overview}, we present a general concept of regression depth. Through a further look at the halfspace regression depth given by \cite{RouHub1999}, we obtain a strong relationship between the halfspace regression depth and Tukey's halfspace depth. This actually shows us a way to extend some other depth notions to the regression setting. A few new regression depth functions are proposed in Section \ref{Extensions}, and illustrated in Section \ref{Illustrations}. The detailed proofs of the main results are provided in the Appendix. Concluding remarks end the paper.

\vskip 0.1 in
\section{Depth functions in the location setting: An overview}
\paragraph{}
\vskip 0.1 in \label{Overview}

In this section, we review the general concept of statistical depth functions, and some of its instances in the location setting.

Let $X\in \mathbb{R}^d$ ($d \ge 1$) be a random vector, $P_X$ be its related probability measure, and $X_1, X_2, \cdots, X_n$ be $n$ i.i.d. random copies of $X$. Following from \cite{ZS2000}, for any $\x \in \mathbb{R}^d$, its depth $\D(\x, P_X)$ with respect to $P_X$ is defined as follows.
\begin{definition}[Zuo and Serfling, 2000b]
\label{Depth}
   Assume that the mapping $\D(\cdot, \cdot): \mathbb{R}^{d}\times \mathcal{P}\rightarrow \mathbb{R}^1$ is bounded, non-negative, and satisfies the following properties P1-P4.
\begin{enumerate}
  \item[P1.] \textbf{Affine-invariance}. For any $d\times d$ nonsingular matrix $\mathbb{A}$ and $d$-vector $\b$, we have that $\D(\mathbb{A}\x + \b, P_{\mathbb{A}X+\b}) = \D(\x, P)$.

  \item[P2.] \textbf{Maximality at $\x_{0}$}. For any $P_X \in \mathcal{P}$, there exists a $\x_0 \in \mathbb{R}^d$ such that: $\D (\x_{0}, P_X) = \sup_{\x \in \mathbb{R}^{d}} \D (\x, P_X)$.

  \item[P3.] \textbf{Monotonicity relate to $\x_0$}. For any $P_X\in \mathcal{P}$ having a deepest point $\x_0$, $\D (\x, P_X) \leq  \D (\x_0 + \lambda  (\x - \x_{0}), P_X)$ holds for $\lambda \in [0, 1]$.

  \item[P4.] \textbf{Vanishing at infinity}. $\D (\x, P_X) \rightarrow 0$ as $\|\x\| \rightarrow \infty$ for each $P_X \in \mathcal{P}$.
\end{enumerate}
\end{definition}

As mentioned in Introduction, classical instances of this general depth notion include: halfspace depth, simplicial depth, projection depth,  Rayleigh depth, and zonoid depth, etc. In the following, we will represent their population version and empirical version. For simplicity, here we denote the random samples as $\mathcal{X}^n = \{X_1, X_2, \cdots, X_n\}$, and the related empirical probability measure as $P_n$.

\emph{Halfspace depth}. Its population version is given in Introduction. Its empirical version is as follows.
\begin{eqnarray}
\label{HD001}
  \HD(\x, P_n) &=& \inf_{\u \in \mathcal{S}^{d - 1}} P_n(\u^\top X \le \u^\top \x)\\
  &=& \inf_{\u \in \mathcal{S}^{d - 1}} \frac{1}{n} \#\left\{i: \u^\top X_i \le \u^\top \x, i = 1, 2, \cdots, n \right\},\nonumber
\end{eqnarray}
where $\#(A)$ denotes the cardinal number of a set $A$. Clearly, $\HD(\x, P_n)$ is equal to the minimum proportion of observations contained in any closed halfspace with $\x$ on its boundary.

\emph{Simplicial depth}. \cite{Liu1990} introduced the following simplicial depth:
\begin{eqnarray*}
  \SD(\x, P) &=& P(\x \in S[X_1, X_2, \cdots, X_{d+1}])\\
  &=& E(I(\x \in S[X_1, X_2, \cdots, X_{d+1}])),
\end{eqnarray*}
i.e., the probability of $\x$ covered by the random simplex $S[X_1, X_2, \cdots, X_{d+1}]$ formed by $d+1$ i.i.d. random vectors $X_1, X_2, \cdots, X_{d+1}$. Its sample version happens to be a $U$-statistics, i.e.,
\begin{eqnarray*}
  \SD(\x, P_n) = {n \choose d+1}^{-1} \sum_{1\leq i_1< i_2 < \cdots< i_{d+1}\leq n} I(\x \in S[X_{i_1}, X_{i_2}, \cdots, X_{i_{d+1}}]).
\end{eqnarray*}

\emph{Projection depth}. \cite{Zuo2003} proposed the following projection depth:
\begin{eqnarray*}
  \PD(\x, P) &=& \frac{1}{1 + O(\x, P)} \text{ with } O(\x, P) = \sup_{\u \in \mathcal{S}^{d-1}} \frac{|\u^\top \x - \text{Med}(P_{\u})|}{\text{MAD}(P_{\u})}.
\end{eqnarray*}
Unlike the aforementioned two type of depths, the definition of projection depth relies on the outlyingness of $\x$ related to $P$. Its sample version is:
\begin{eqnarray*}
  \PD(\x, P_n) = \frac{1}{1 + O(\x, P_n)} \text{ with } O(\x, P_n) = \sup_{\u \in \mathcal{S}^{d-1}} \frac{|\u^\top \x - \text{Med}(P_{\u, n})|}{\text{MAD}(P_{\u, n})}.
\end{eqnarray*}
Here $P_{\u}$ and $P_{\u, n}$ denote the one dimensional probability measure and empirical probability measure related to $\u^\top X$ and $\{\u^\top X_1, \u^\top X_2, \cdots, \u^\top X_n\}$, respectively. MAD$(P_{\u, n})$ denotes the median absolute deviation of $\u^\top X_i$, $i = 1, 2, \cdots, n$, to their median $\text{Med}(P_{\u, n})$, which is specified in \eqref{med}.

\emph{Rayleigh depth}. This depth is actually a hybrid of the projection depth by replacing (Med, MAD) with $(\mu, \sigma)= $ (mean, standard deviation) \citep{Zuo2003}, namely,
\begin{eqnarray*}
  \PD^r(\x, P) &=& \frac{1}{1 + O^r(\x, P)} \text{ with } O^r(\x, P) = \sup_{\u \in \mathcal{S}^{d-1}} \frac{|\u^\top \x - \mu(P_{\u})|}{\sigma(P_{\u})}.
\end{eqnarray*}
Since the solution of $O^r(\x, P)$ here is that of a Rayleigh quotient problem \citep{HWW2011}, in the sequel we call it \emph{Rayleigh depth} for convenience. Its sample version is:
\begin{eqnarray*}
  \PD^r(\x, P_n) = \frac{1}{1 + O^r(\x, P_n)} \text{ with } O^r(\x, P_n) = \sup_{\u \in \mathcal{S}^{d-1}} \frac{|\u^\top \x - \mu(P_{\u, n})|}{\sigma(P_{\u, n})}.
\end{eqnarray*}

\emph{Zonoid depth}. This depth was introduced by \cite{KM1997}. Different from all depth above, zonoid depth is of $L_2$ nature. Its population version is:
\begin{eqnarray*}
  \ZD(\x, P) &=& \sup\Bigg{\{}\alpha: \x = \int_{\mathbb{R}^d} g(X) X dP, ~\int_{\mathbb{R}^d} g(X) dP = 1, \\
             & & \quad\quad\quad\quad   g(\cdot) \text{ is measurable, and } g(\t) \in [0, 1/\alpha] \text{ for any } \t \in \mathbb{R}^d \Bigg{\}}.
\end{eqnarray*}
Its sample version is then:
\begin{eqnarray*}
  \ZD(\x, P_n) &=& \sup\Bigg{\{}\alpha: \x = \sum_{i=1}^n \lambda_i X_i, ~\sum_{i=1}^n \lambda_i = 1, ~n\lambda_i \in [0, 1/\alpha], ~\forall i \Bigg{\}},
\end{eqnarray*}
if $\x \in$ the convex hull of $\mathcal{X}^n$. Otherwise, its value is defined to be zero.

All instances mentioned above satisfy Properties P1-P4 of defining a general depth function, but for the \emph{sample simplicial depth, who does not satisfies Property P3}.

\vskip 0.1 in
\section{General notions of regression depth}
\paragraph{}
\vskip 0.1 in \label{HRD}

In this section, we will first provide a general definition of regression depth, and then give a further look at the halfspace regression depth proposed by \cite{RouHub1999}. All detailed proofs of the main results can be found in the Appendix.

\subsection{General regression depth function}\label{GeneralRD}
\paragraph{}

Suppose the i.i.d. random samples $\mathcal{Z}^n := \{(X_i, Y_i)\}_{i=1}^n \subset \mathbb{R}^d \times \mathbb{R}^1$ are generated from the following linear model:
\begin{eqnarray}
\label{linearM}
    Y = \beta_{0} + \bm{\beta}_{1}^\top X + \varepsilon,
\end{eqnarray}
where $Y$ denotes the response variable, $X$ the $d$-variate covariates, $\bm{\theta} = (\beta_0, \bm{\beta}_1^\top )^\top$ the unknown parameter, and $\varepsilon$ the random model error, which is symmetrically distributed about 0 conditionally on given $X$.

For simplicity, we denote $\bm\theta_0$ as the true value of $\bm\theta$, write $\z = (\x^\top, \y)^\top $, $\w = (1, \x^\top )^\top $, $Z = (X^\top, Y)^\top $, $W = (1, X^\top )^\top $, $\mathcal{S}^{d} = \{\u \in \mathbb{R}^{d+1} : \|\u\| = 1\}$, and let $\bar{P} = P_{X, Y}$ be the probability measure related to $Z$. Hereafter,  we assume that the covariance matrix of $W$ to be positive, which further implies $E(WW^\top)$ is positive.

Similar to \cite{ZS2000}, for a given class of probability measures $\mathcal{F}$, we can define the
general regression depth function as follows.

\begin{definition}
\label{RegDepth}
   Assume that the mapping $\RD(\cdot, \cdot): \mathbb{R}^{d+1}\times \mathcal{F}\rightarrow [0, 1]$ satisfies the following properties Q1-Q4.
\begin{enumerate}
  \item[Q1.] \textbf{Affine and scale equivariance} \citep{RouLer1987}. Namely, $\RD(\cdot, \cdot)$ satisfies
$$\RD\left({\bm\beta_{0}\choose \Sigma^{-1} \bm\beta_{1}}, P_{\Sigma X, Y}\right) = \RD\left(\bm\theta, P_{X, Y}\right)$$
and
$$\RD(b\bm\theta, P_{X, bY}) = \RD(\bm\theta, P_{X, Y})$$
for nonsingular $(d - 1)\times (d - 1)$ matrices $\Sigma$ and $b \in \mathbb{R}^1$, respectively.

  \item[Q2.] \textbf{Maximality at $\bm\theta_{0}$}. For any $P_{X, Y} \in \mathcal{F}$, we have that: $\RD (\bm\theta_{0}, P_{X, Y}) = \sup_{\bm\theta \in \mathbb{R}^{d}} \RD (\bm\theta, P_{X, Y})$.

  \item[Q3.] \textbf{Monotonicity relate to $\bm\theta_0$}. For any $P_{X, Y}\in \mathcal{F}$, $\RD (\bm\theta, P_{X, Y}) \leq  \RD (\bm\theta_0 + \lambda  (\bm\theta - \bm\theta_{0}), P_{X, Y})$ holds for $\lambda \in [0, 1]$.

  \item[Q4.] \textbf{Vanishing at infinity}. $\RD (\bm\theta, P_{X, Y}) \rightarrow 0$ as $\|\bm\theta\| \rightarrow \infty$ for each $P_{X, Y} \in \mathcal{F}$.
\end{enumerate}
\end{definition}

Then $\RD(\cdot, \cdot)$ is called a \emph{statistical regression depth function}. Property Q1 is important, because it requires the depth value of the coefficient $\bm\theta$ to be independent of some common data transformations, such as rotation and rescaling. Q2 and Q4 ensure that one can define median-like estimators relying the regression depth function. Q3 ensures a decreasing ordering if $\bm\theta$ moves away from $\bm\theta_0$ to infinity along any line stemming from $\bm\theta_0$.

\subsection{A further look at the halfspace regression depth}
\paragraph{}

It is known that \cite{RouHub1999} developed a regression depth function extended from Tukey's halfspace depth in the location setting. By \cite{BH1999}, its sample version can actually be expressed as follows:
\begin{eqnarray}
\label{HRD001}
   && \HRD(\bm\theta, \bar{P}_n)\\
   && \quad=\inf_{v_0\in \mathcal{R}^1, \vv_1 \in \mathcal{S}^{d-1}} \frac{1}{n} \min\left\{\sum_{i=1}^n I\left(r_i(\bm\theta) (\vv_1^\top X_i - v_0) \ge 0\right), ~ \sum_{i=1}^n I\left(r_i(\bm\theta) (\vv_1^\top X_i - v_0) \le 0\right)\right\},\nonumber
\end{eqnarray}
where $r_i(\bm\theta) = Y_i - \beta_0 - \bm\beta_1^\top X_i = Y_i - W_i^\top \bm\theta$ and $\bar{P}_n$ denotes the empirical probability measure related to $\mathcal{Z}^n$. But \emph{slightly differently}, the discussions in \cite{RouHub1999, BH1999} are based on an integer valued function, i.e., $rdepth(\bm\theta, \bar{P}_n) = n \HRD(\bm\theta, \bar{P}_n)$, instead. Here we technically use $\HRD(\bm\theta, \bar{P}_n)$ to ensure that $\HRD(\bm\theta, \bar{P}_n) \in [0, 1]$. Furthermore, we use $\sum_{i=1}^n I\left(r_i(\bm\theta) (\vv_1^\top X_i - v_0) \le 0\right)$ instead of $\sum_{i=1}^n I\left(r_i(\bm\theta) (\vv_1^\top X_i - v_0) < 0\right)$. But it would make no significant difference because $\sum_{i=1}^n I\left(r_i(\bm\theta) (\vv_1^\top X_i - v_0) \le 0\right) = n - \sum_{i=1}^n I\left(r_i(\bm\theta) (\vv_1^\top X_i - v_0) > 0\right)$, and so is the case for $\sum_{i=1}^n I\left(r_i(\bm\theta) (\vv_1^\top X_i - v_0) \ge 0\right)$.

Observe that the connection between \eqref{HRD001} and \eqref{HD001} is not so obvious. Hence, we propose to consider the following proposition.

\begin{proposition}
\label{equivalenceHD}
  Suppose $\mathcal{Z}^n$ are generated from \eqref{linearM}. Then the halfspace regression depth of $\bm\theta$ with respect to $\mathcal{Z}^n$ is equal to the halfspace depth of \textbf{0} with respect to $\{(Y_i - W_i^\top \bm\theta) W_i\}_{i=1}^n$. That is,
  \begin{eqnarray*}
    \HRD(\bm\theta, \bar{P}_n) = \HD(\0, \bar{P}_{\theta, n}),
  \end{eqnarray*}
  where $\bar{P}_{\theta, n}$ denotes the empirical probability measure related to $\{(Y_i - W_i^\top \bm\theta) W_i\}_{i=1}^n$.
\end{proposition}

Based on this result, it is easy to derive that, for any given $\bm\theta$, $\HRD(\bm\theta, \bar{P}_n)$ converges in probability to its population version, i.e.,
\begin{eqnarray*}
  \HRD(\bm\theta, \bar{P}) = \inf_{\u\in \mathcal{S}^d} P\left(\u^\top(Y - W^\top \bm\theta)W \le 0\right) := \HD(\0, \bar{P}_{\theta})
\end{eqnarray*}
by using the empirical process theory similarly to \cite{BH1999} and \cite{Zuo2003}. We omit the details.

Intuitively, $\HRD(\bm\theta, \bar{P}) = \HD(\0, \bar{P}_{\theta})$ explains the reasonability of the definition of the halfspace regression depth defined in \cite{RouHub1999} from an another point of view. In fact, if $E(\varepsilon|X) = 0$, then $E((Y - W^\top \bm\theta_0)W) = E(\varepsilon W) = \0$, which implies that $(Y - W^\top \bm\theta_0)W$ is distributed around $\0$. Hence, $\HRD(\bm\theta_0, \bar{P})$ tends to take a large value. On the other hand, for any $\bm\theta \neq \bm\theta_0$,  observe that $E(\u_{\theta}^\top (Y - W^\top \bm\theta)W) = \u_{\theta}^\top E(WW^\top) (\bm\theta_0 - \bm\theta) > 0$ if $E(WW^\top)$ is positive, where $\u_{\theta} = (\bm\theta_0 - \bm\theta) / \|\bm\theta_0 - \bm\theta\|$. Hence, $\HRD(\bm\theta, \bar{P})$ tends to take a small value due to $\HRD(\bm\theta, \bar{P}) \leq P\left(\u_{\theta}^\top(Y - W^\top \bm\theta)W \le 0\right)$.

Since the halfspace regression depth of \cite{RouHub1999} is the first regression depth notion in the literature, a natural question that arises is: Whether or not it satisfies all four Properties Q1-Q4 of defining a general regression depth function given in Section \ref{GeneralRD}?

Before answering this question, we need first to clarify a concept, i.e., \emph{regression halfspace symmetry}, specified as follows.

\begin{definition}
  We say that $Z = (X^\top, Y)^\top $ is regression halfspace symmetrically distributed about $\bm\theta_0 \in \mathbb{R}^{d+1}$ if $\bm\theta_0$ satisfies
  \begin{eqnarray*}
     P\left(\u^\top (Y - W^\top \bm\theta_0)W \le 0 \right) \ge \frac{1}{2}, \quad \text{for } \forall \u \in \mathcal{S}^d.
  \end{eqnarray*}
\end{definition}

\begin{remark}
  The statement that $(X^\top, Y)^\top$ is regression halfspace symmetrically distributed is actually equivalent to that there is a $\bm\theta_0 \in \mathbb{R}^{d+1}$ such that $(Y - W^\top \bm\theta_0)W$ is halfspace symmetrically distributed about \textbf{0}; for the definition of halfspace symmetry, we refer reads to \cite{ZS2000c}.
\end{remark}

Bearing this concept in mind, we obtain the following result.

\begin{theorem}
\label{ThHRD}
  Suppose $X$ is continuous and $Z = (X^\top, Y)^\top $ is regression halfspace symmetrically distributed about a unique $\bm\theta_0$. The halfspace regression depth function $\HRD(\bm\theta, \bar{P})$ satisfies Properties Q1, Q2 and Q4.
\end{theorem}

\begin{remark}
\label{remark2}
In practice, the assumption of regression halfspace symmetry about a unique $\bm\theta_0$ is not too rigorous. It is easy to check that: if $\varepsilon$ is symmetrically distributed about 0 conditionally on a continuous covariate $X$, then this assumption holds; see the Appendix for a detailed proof.
\end{remark}

Unfortunately, it is worth mentioning that we are unable to check that whether or not the halfspace depth function satisfies Properties Q3 without further assumptions. This implies that we can not claim a decreasing ordering from the center to outside induced from this depth similar to the location setting. However, Properties Q1, Q2 and Q4 are still useful if the main purpose is to develop an affine equivariant median-like regression estimator for the unknown coefficient with positive breakdown point robustness.

\section{Some other regression depth functions}\label{Extensions}
\paragraph{}

From Section \ref{Overview}, it is readily to see that many depths in location setting are motivated by Tukey's halfspace depth. On the other hand, the halfspace regression depth of $\bm\theta$ with respect to $\mathcal{Z}^n$ proposed by \cite{RouHub1999} is equal to the halfspace depth of $\0$ with respect to $\{(Y_i - W_i^\top \bm\theta)W_i)\}_{i=t}^n$. Motivated by this, we rewrite the form of the simplicial regression depth given in \cite{RouHub1999}, extend the usually projection depth, Rayleigh depth and zonoid depth into the regression setting, and then check that whether or not they satisfy all four Properties Q1-Q4 given in Section \ref{GeneralRD}. The proofs are given in the Appendix.

\textbf{\emph{Simplicial regression depth}}. For given $\mathcal{Z}^n$, similar to Proposition \ref{equivalenceHD}, we rewrite the simplicial regression depth of $\bm\theta$ with respect to $\mathcal{Z}^n$ given in \cite{RouHub1999} in the following form:
\begin{eqnarray*}
  \SRD(\bm\theta, \bar{P}_n) &=& \SD(\0, \bar{P}_{\theta, n}).
\end{eqnarray*}

Using the central limit theory for $U$-statistics, it is easy to check that, for given $\bm\theta$, $\SRD(\bm\theta, \bar{P}_n)$ converges in distribution and hence in probability to
\begin{eqnarray*}
  \SRD(\bm\theta, \bar{P}) = \bar{P}(\0 \in S[r_{1}(\bm\theta)W_{1}, \cdots, r_{d+1}(\bm\theta)W_{d+1}]).
\end{eqnarray*}

For $\SRD(\bm\theta, \bar{P})$, we have the following result.

\begin{theorem}
\label{th:SRD}
  Under the same conditions of Theorem \ref{ThHRD}, the simplicial regression depth function $\SRD(\bm\theta, \bar{P})$ satisfies Properties Q1, Q2 and Q4.
\end{theorem}

Once again, the check of whether or not simplicial regression depth satisfying Property Q3 is difficult to achieve without further assumptions.

\textbf{\emph{Projection regression depth}}. Similarly, we may define the projection regression depth of $\bm\theta$ with respect to $\mathcal{Z}^n$ as
\begin{eqnarray*}
  \PRD^*(\bm\theta, \bar{P}_n) &=& \PD(\0, \bar{P}_{\theta, n}).
\end{eqnarray*}
However, note that the denominator part in the outlyingness function contains the coefficient $\bm\theta$. It may bring inconvenience to its theoretical derivation. Hence, we need to avoid this.

Observe that the main purpose of adding the denominator part to the outlyingness function is to make the depth function to be scale equivariant. Hence, we modify the definition of projection regression depth slightly as follows. That is, we propose to consider instead the following projection regression depth:
\begin{eqnarray*}
  \PRD(\bm\theta, \bar{P}_n) = \frac{1}{1 + \tilde{O}(\bm\theta, \bar{P}_n)}
\end{eqnarray*}
with
\begin{eqnarray*}
   \tilde{O}(\bm\theta, \bar{P}_n) = \sup_{\u\in \mathcal{S}^d} \frac{\left|\text{Med}(\bar{P}_{\u, \theta, n})\right|}{\text{MAD}(\tilde{P}_{\u, n})} = \sup_{\u\in \mathcal{S}^d} \frac{\left|\u^\top \0 - \text{Med}(\bar{P}_{\u, \theta, n})\right|}{\text{MAD}(\tilde{P}_{\u, n})},
\end{eqnarray*}
where $\bar{P}_{\u, \theta, n}$ and $\tilde{P}_{\u, n}$ denotes the empirical probability measure related to $\{(Y_i - W_i^\top \bm\theta) \u^\top W_i\}_{i=1}^n$ and $\{\u^\top Y_iW_i\}_{i=1}^n$, respectively.

Under some regular conditions similar to \cite{Zuo2003}, we can derive that $\PRD(\bm\theta, \bar{P}_n)$ converges in probability to its population version as follows:
\begin{eqnarray*}
  \PRD(\bm\theta, \bar{P}) = \frac{1}{1 + \tilde{O}(\bm\theta, \bar{P})}, \text{ with }
   \tilde{O}(\bm\theta, \bar{P}) = \sup_{\u\in \mathcal{S}^d} \frac{\left|\text{Med}(\bar{P}_{\u, \theta})\right|}{\text{MAD}(\tilde{P}_{\u})}.
\end{eqnarray*}

About $\PRD(\bm\theta, \bar{P})$, we have the following result.

\begin{theorem}
\label{th:PRD}
  Under the same conditions of Theorem \ref{ThHRD}, the projection regression depth function $\PRD(\bm\theta, \bar{P})$ satisfies Properties Q1, Q2 and Q4.
\end{theorem}

\textbf{\emph{Rayleigh regression depth}}. Similar to the projection regression depth, we define the sample Rayleigh regression depth as
\begin{eqnarray*}
  \PRD^r(\bm\theta, \bar{P}_n) = \frac{1}{1 + \tilde{O}^r(\theta, \bar{P}_n)}
\end{eqnarray*}
with
\begin{eqnarray*}
   \tilde{O}(\bm\theta, \bar{P}_n) = \sup_{\u\in \mathcal{S}^d} \frac{\left|\mu(\bar{P}_{\u, \theta, n})\right|}{\sigma(\tilde{P}_{\u, n})} = \sup_{\u\in \mathcal{S}^d} \frac{\left|\u^\top \0 - \mu(\bar{P}_{\u, \theta, n})\right|}{\sigma(\tilde{P}_{\u, n})}.
\end{eqnarray*}

Correspondingly, its population version is
\begin{eqnarray*}
  \PRD^r(\bm\theta, \bar{P}) = \frac{1}{1 + \tilde{O}^r(\bm\theta, \bar{P})}
\end{eqnarray*}
with
\begin{eqnarray*}
   \tilde{O}(\bm\theta, \bar{P}) = \sup_{\u\in \mathcal{S}^d} \frac{\left|\mu(\bar{P}_{\u, \theta})\right|}{\sigma(\tilde{P}_{\u})} = \sup_{\u \in\mathcal{S}^d} \frac{|E\left(\u^\top (Y - W^\top\theta)W\right)|}{\sqrt{\text{Var}(\u^\top YW)}}.
\end{eqnarray*}

\begin{theorem}
\label{th:RRD}
  Suppose $E(\varepsilon|X) = 0$ and $\bar{\Sigma} = E(WW^\top)$ is positive. Then Rayleigh regression depth satisfies all four properties of Definition \ref{RegDepth}.
\end{theorem}

Theorem \ref{th:RRD} indicates that Rayleigh regression depth can serve as the usual depth function in the location setting to provide a fully center-outward ordering for the coefficient parameters.

\textbf{\emph{Zonoid regression depth}}. Similar to the discussions above, the sample zonoid regression depth can be given as
\begin{eqnarray*}
  \ZRD(\bm\theta, \bar{P}_n) = \ZD(\0, \bar{P}_{\theta, n}),
\end{eqnarray*}
and its population version is
\begin{eqnarray*}
  \ZRD(\bm\theta, \bar{P}) = \ZD(\0, \bar{P}_{\theta}).
\end{eqnarray*}

\begin{theorem}
\label{th:ZRD}
  Suppose $E(\varepsilon|X) = 0$ and $\bar{\Sigma} = E(WW^\top)$ is positive. Then zonoid regression depth satisfies Properties Q1, Q2, and Q3.
\end{theorem}

In fact, when proving Q2 in Theorem \ref{th:ZRD}, it is easy to check that the maximizer of the sample zonoid regression depth $\ZRD(\bm\theta, \bar{P}_n)$ is the conventional least squares estimator $\hat{\bm\theta} = (\sum_{i=1}^n W_i W_i^\top)^{-1} \sum_{i=1}^n W_i Y_i$ with $\ZRD(\hat{\bm\theta}, \bar{P}_n) = 1$.

\vskip 0.1 in
\section{Illustrations}\label{Illustrations}
\paragraph{}
\vskip 0.1 in

To gain more insight into various regression depth notions mentioned above, we provide some illustrations in this section. The data set is generated from the following linear model:
\begin{eqnarray*}
  Y = 0.5 + 0.5 X + \varepsilon,
\end{eqnarray*}
where $X \sim N(0, 1)$ and $\varepsilon \sim N(0, 0.2)$. The sample size $n$ is 300. Its scatter plot is given in Figure~\ref{fig:scatter}.

\begin{figure}[H]
\centering
	\includegraphics[angle=0,width=3.5in]{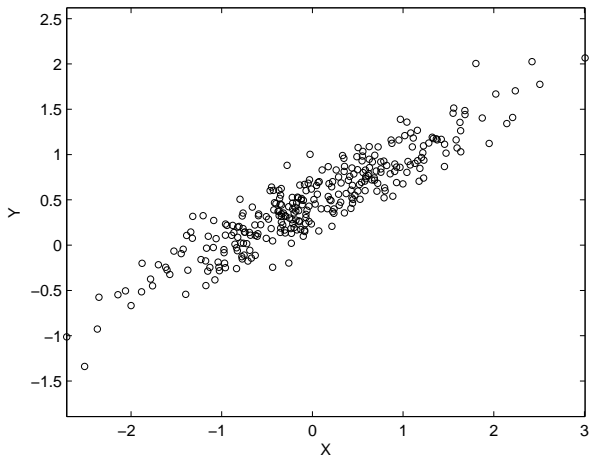}
\caption{Shown is the scatter plot of the data set.}
\label{fig:scatter}
\end{figure}

For each notion of the aforementioned regression depths, we plot two figures, namely, a 3-dimensional depth plot and its corresponding contours over $\bm\theta = (\beta_0, \beta_1)^\top \in [0, 1] \times [0, 1]$; see Figures~\ref{fig:HRD} - \ref{fig:ZRD}. Seven contours, as well a maximizer of each regression depth function (see the big point in the center), are reported. The figures shows that, except the sample Rayleigh regression depth, all reported depth contours of these four depth functions are not convex, but still roughly nested, and their shape is data-dependent. Generally speaking, among these four depths, the contours of projection regression depth appears to be smoother than those of the halfspace and simplicial regression depth, but rougher than those of the zonoid regression depth. Furthermore, the roughly nested construction of contours indicates that the value of all regression depths tends to become smaller when $\bm\theta$ is moving from the center to outside, although it may be not strictly decreasing for, e.g., halfspace regression depth. One exception is the Rayleigh regression depth. Its sample contours are some nested ellipses, which coincides with the Theorem \ref{th:RRD}.

\begin{figure}[H]
\centering
	\subfigure[Depth plot]{
	\includegraphics[angle=0,width=2.9in]{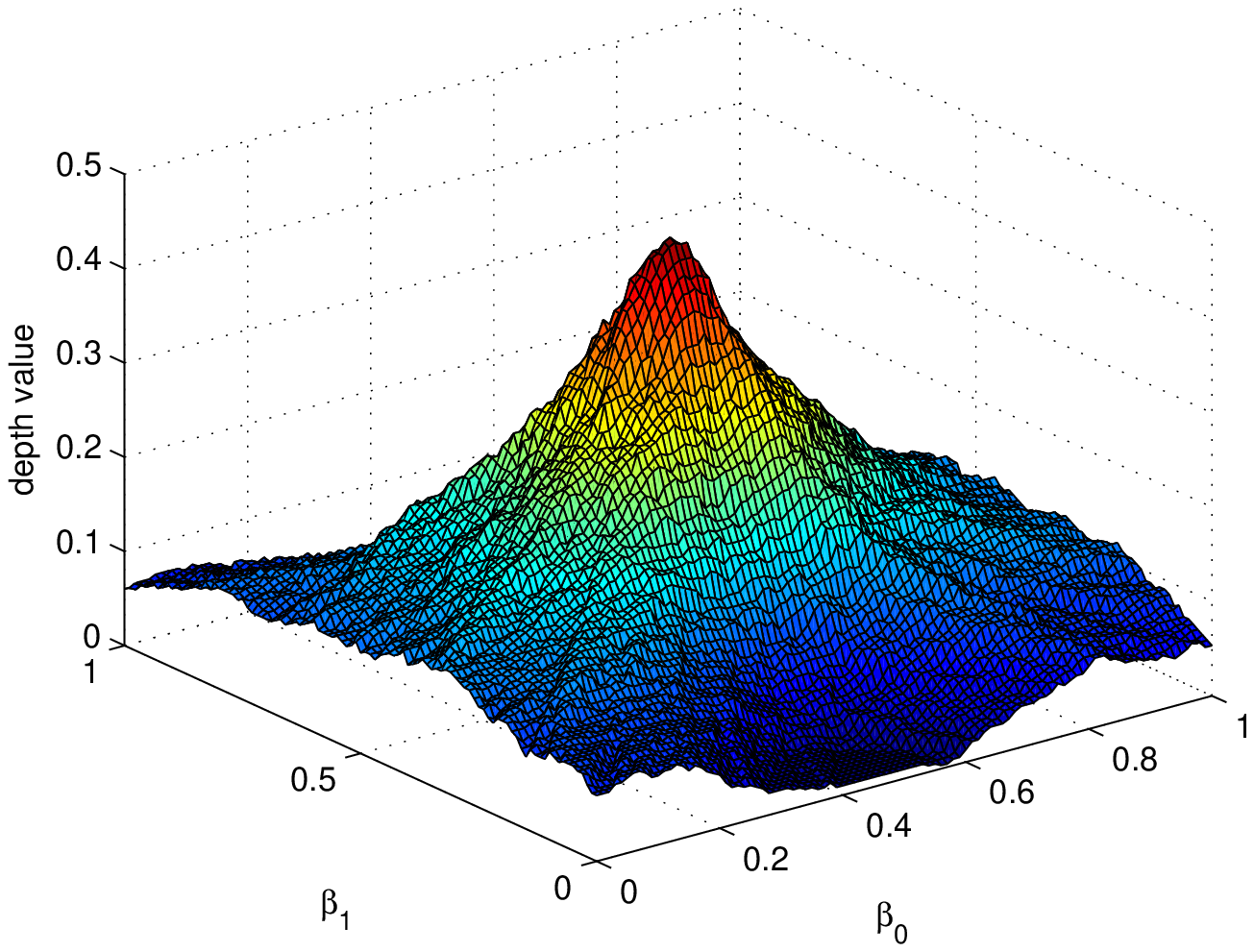}
	}\quad
	\subfigure[Contours]{
	\includegraphics[angle=0,width=2.9in]{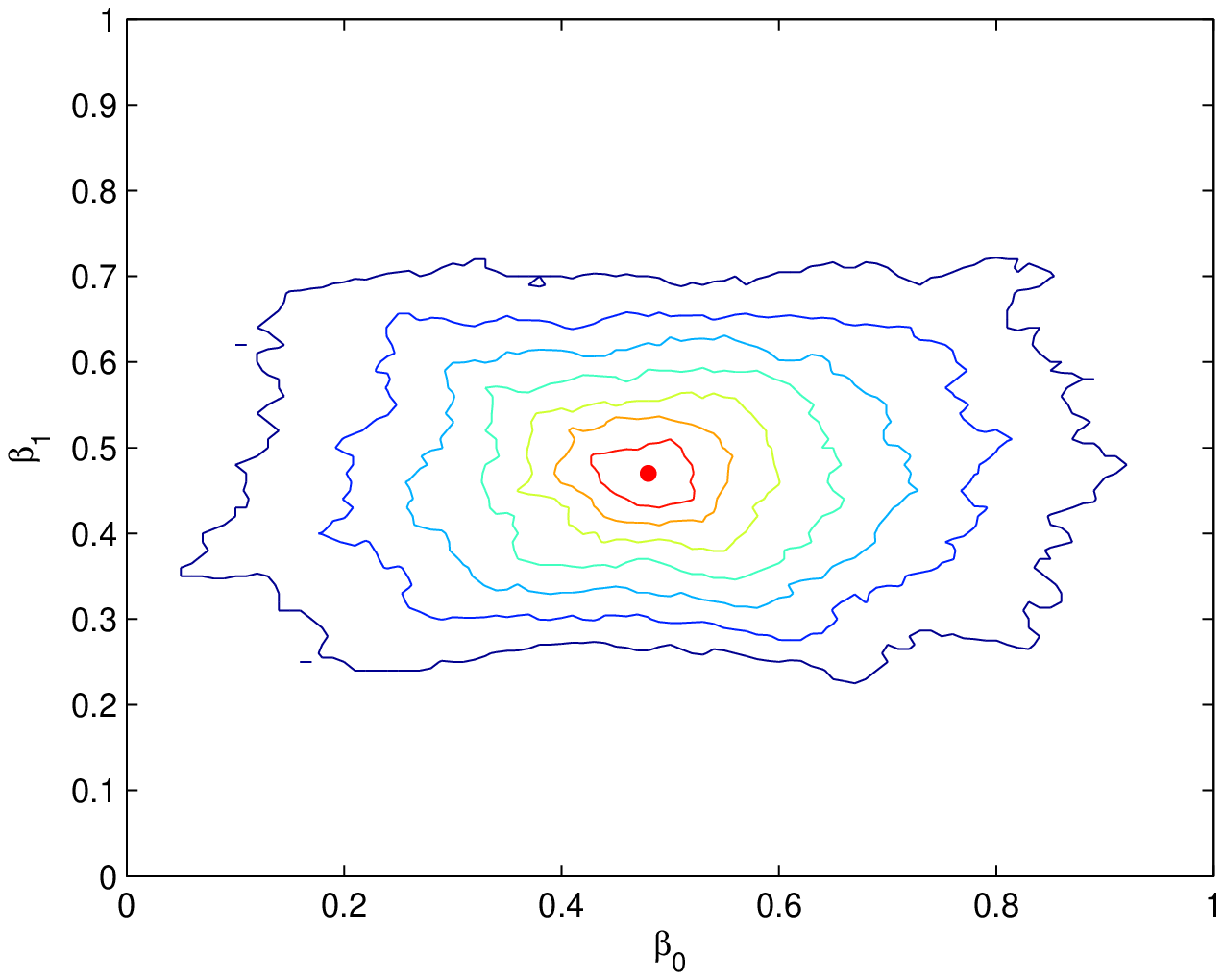}
	}
\caption{Shown are illustrations for the sample halfspace regression depth function. The depth values of sever contours are $\alpha = 0.1500, 0.1956, 0.2411, 0.2867, 0.3323, 0.3779, 0.4234$ from the periphery inwards.}
\label{fig:HRD}
\end{figure}

\begin{figure}[H]
\centering
	\subfigure[Depth plot]{
	\includegraphics[angle=0,width=2.9in]{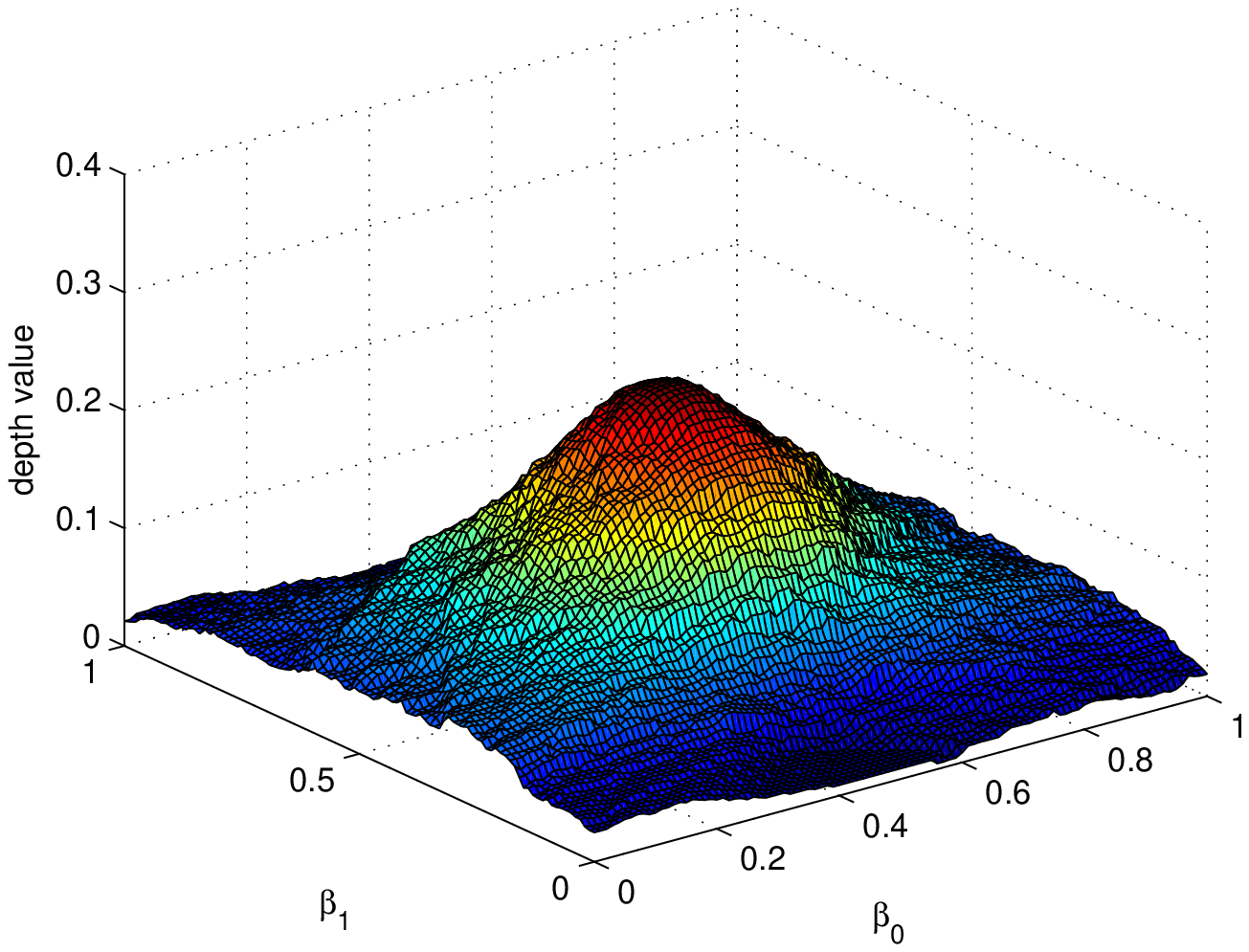}
	}\quad
	\subfigure[Contours]{
	\includegraphics[angle=0,width=2.9in]{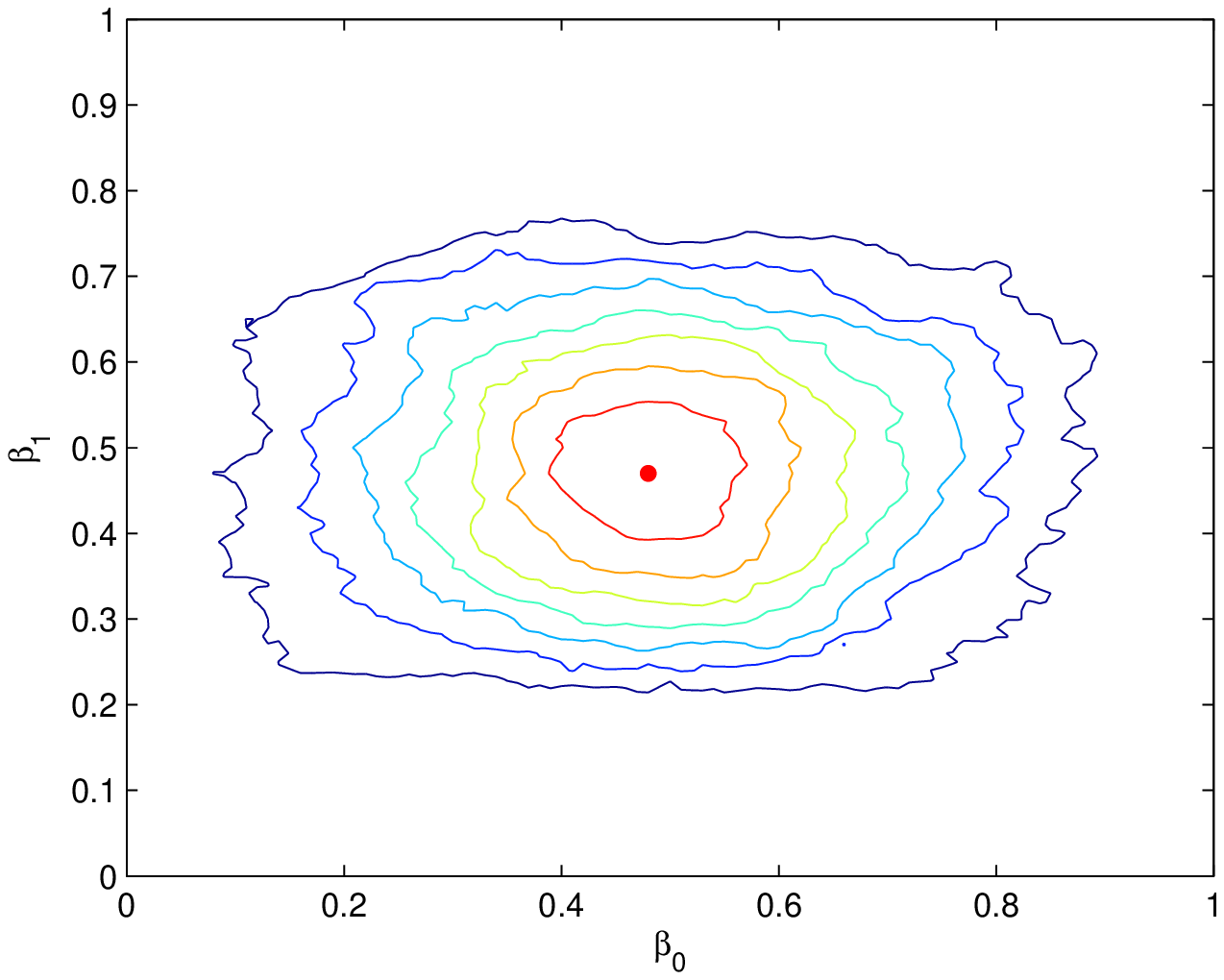}
	}
\caption{Shown are illustrations for the sample simplicial regression depth function. The depth values of sever contours are $\alpha = 0.0800, 0.1046, 0.1292, 0.1538, 0.1784, 0.2030, 0.2276$ from the periphery inwards.}
\label{fig:SRD}
\end{figure}

\begin{figure}[H]
\centering
	\subfigure[Depth plot]{
	\includegraphics[angle=0,width=2.9in]{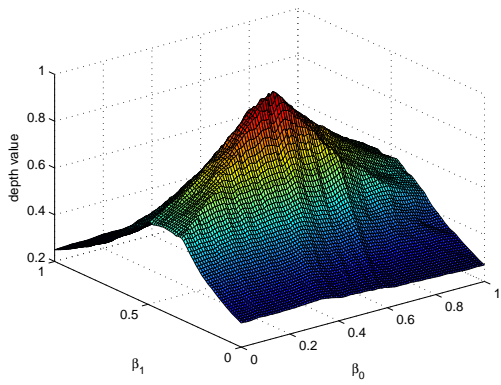}
	}\quad
	\subfigure[Contours]{
	\includegraphics[angle=0,width=2.9in]{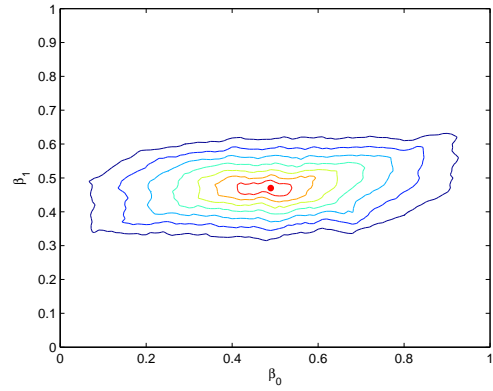}
	}
\caption{Shown are illustrations for the sample projection regression depth function. The depth values of sever contours are $\alpha = 0.6000, 0.6546, 0.7092, 0.7637, 0.8183, 0.8729, 0.9275$ from the periphery inwards.}
\label{fig:PRD}
\end{figure}

\begin{figure}[H]
\centering
	\subfigure[Depth plot]{
	\includegraphics[angle=0,width=2.9in]{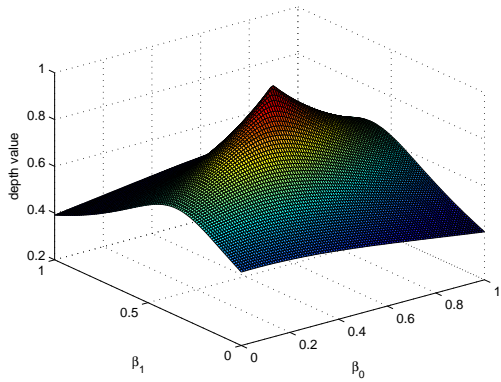}
	}\quad
	\subfigure[Contours]{
	\includegraphics[angle=0,width=2.9in]{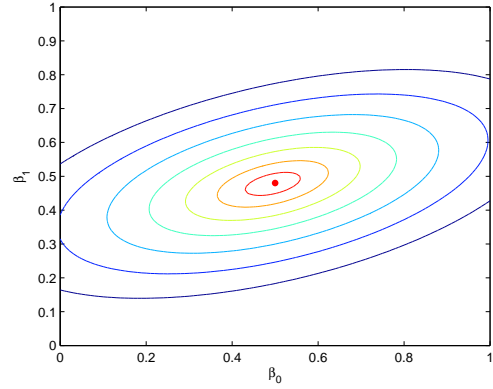}
	}
\caption{Shown are illustrations for the sample Rayleigh regression depth function. The depth values of sever contours are $\alpha = 0.6000, 0.6561, 0.7123, 0.7684, 0.8245, 0.8807, 0.9368$ from the periphery inwards.}
\label{fig:RRD}
\end{figure}

\begin{figure}[H]
\centering
	\subfigure[Depth plot]{
	\includegraphics[angle=0,width=2.9in]{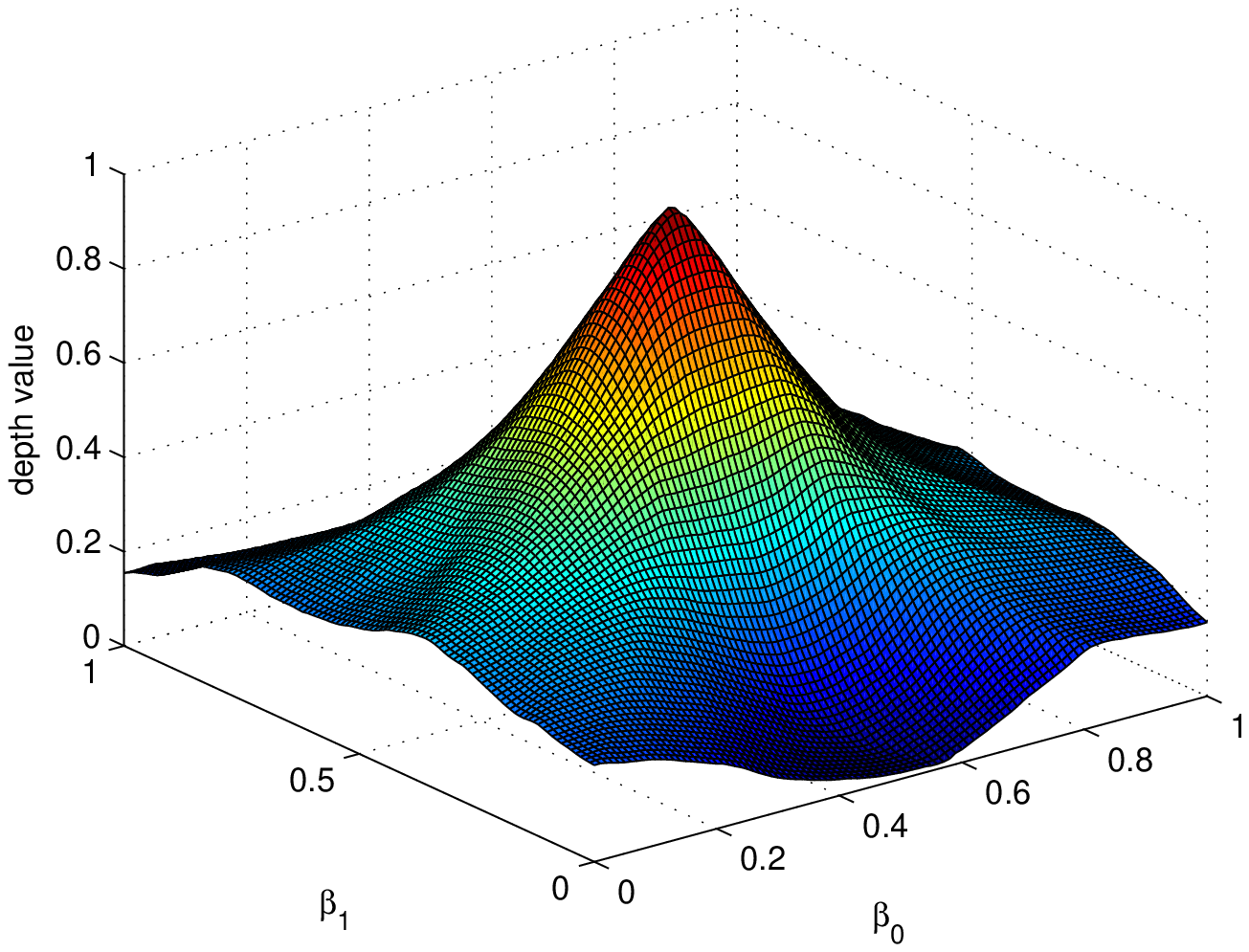}
	}\quad
	\subfigure[Contours]{
	\includegraphics[angle=0,width=2.9in]{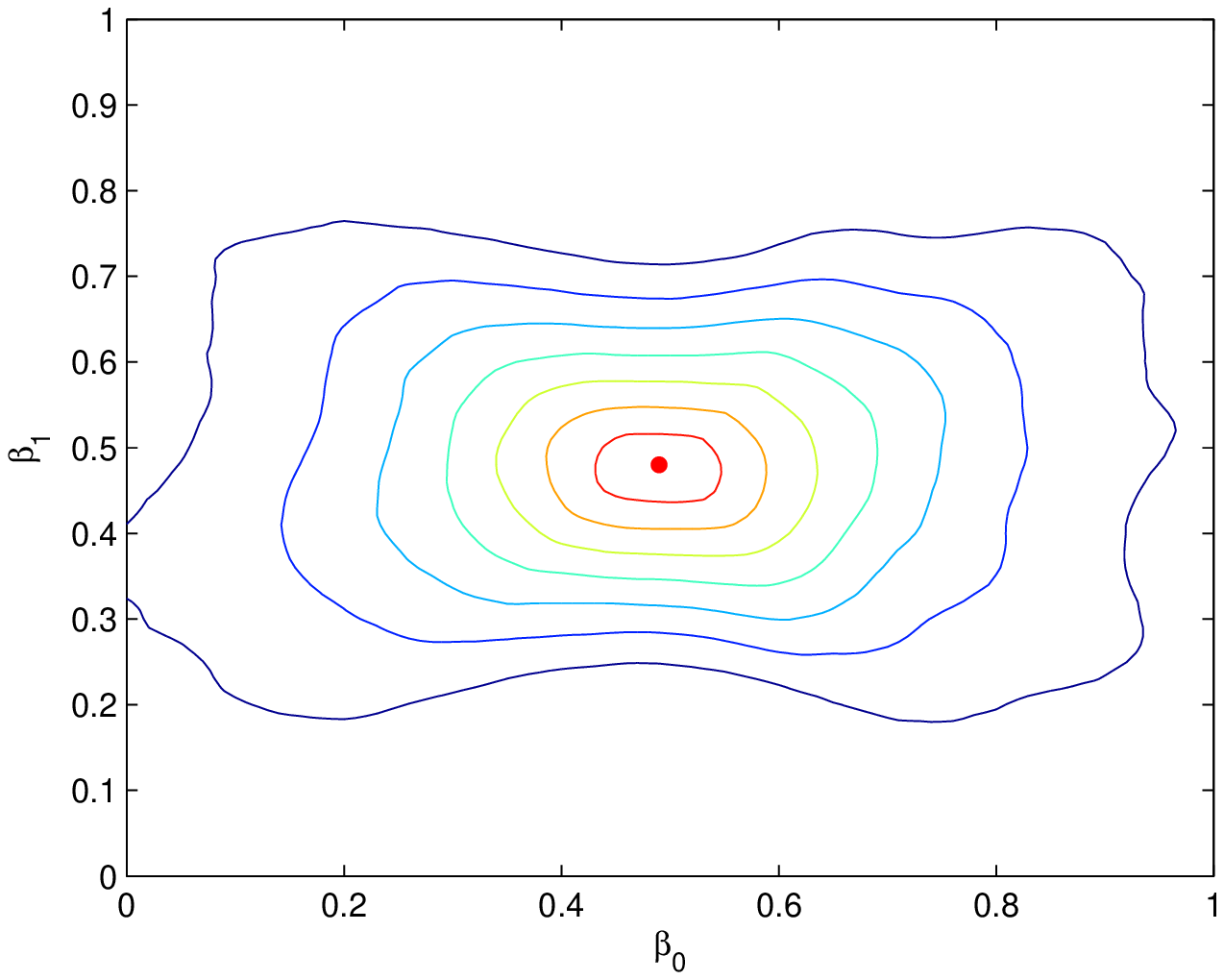}
	}
\caption{Shown are illustrations for the sample zonoid regression depth function. The depth values of sever contours are $\alpha = 0.3000, 0.3992, 0.4985, 0.5977, 0.6970, 0.7962, 0.8954$ from the periphery inwards.}
\label{fig:ZRD}
\end{figure}

Furthermore, all maximizers of these five regression depth functions are very closed to the true coefficient parameter $\bm\theta_0 = (0.5, 0.5)^\top$ in this example. It indicates that, relying on any of them, we are able to develop a proper median-like estimator for the unknown regression efficient.

\vskip 0.1 in
\section{Concluding remarks}\label{Remark}
\paragraph{}
\vskip 0.1 in

In this paper, we investigated the possibility of generalizing some usual depths into the regression setting. A general concept of regression depth was proposed. Several new regression depth notions, such as projection regression depth, were suggested relying on a further look at the relationship between the halfspace regression depth of \cite{RouHub1999} and Tukey's location depth.

Unfortunately, it turns out that the halfspace, simplicial, projection, zonoid regression do not satisfy Q3 of \ref{RegDepth}. Hence, using these four depth notions, it is probably impossible to induce a strictly decreasing ordering when the parameter $\bm\theta$ is moving along any line stemming from the center $\bm\theta_0$ to outside.

\emph{However}, the results of Theorems \ref{ThHRD}-\ref{th:ZRD}, as well as the illustrations given in Section \ref{Illustrations}, indicate that it is still possible to induce some median-like estimators relying on these depth nations. In the literature, it is well known that the breakdown point robustness of the projection median is much higher than that of Tukey's halfspace median \citep{Zuo2003}. Hence, we anticipate that this is also the case for projection regression depth, and will pursuit it in the future. On the other hand, the Rayleigh regression depth satisfies Property Q3, confirmed by Figure~\ref{fig:RRD}. In this sense, if the main purpose is to provide a full ordering like in the location setting, one may use the Rayleigh regression depth.

Observe that the computational issue is usually concerned by practitioners related to statistical depth functions. Fortunately, the depth value of all regression depth notions mentioned in this paper is computable, because their definitions are based on their counterparts in the location setting, thanking to the latest developments of the computation of halfspace depth \citep{LZ2014a, DM2014}, projection depth \citep{LZ2014b}, and zonoid depth \citep{MLB2009}, and so on. We refer reads to \cite{PMD2016} for details of a \textbf{R} package \emph{ddalpha}, which includes implementations of most depth notions in the location setting.

\vskip 0.1 in
\section*{Acknowledgements}
\paragraph{}
\vskip 0.1 in \label{Conclusion}

The research is supported by NNSF of China (Grant No.11601197, 11461029), China Postdoctoral Science Foundation funded project (2016M600511, 2017T100475),  NSF of Jiangxi Province (No.20171ACB21030, 20161BAB201024), and the Key Science Fund Project of Jiangxi provincial education department (No.GJJ150439).
\medskip

\vskip 0.1 in
\section*{Appendix: Detailed proofs of the main results}\label{Appendix}
\paragraph{}
\vskip 0.1 in

In this appendix, we provided the detailed proofs of the main proposition and theorems.

\begin{proof}[Proof of Proposition \ref{equivalenceHD}]
For the halfspace regression depth given above, a simple derivation leads to
\begin{eqnarray*}
  && \HRD(\bm\theta, \bar{P}_n)\\
   && \quad=\inf_{v_0\in \mathbb{R}^1, \vv_1 \in \mathcal{S}^{d-1}} \frac{1}{n} \min\left\{\sum_{i=1}^n I\left(r_i(\bm\theta) (\vv_1^\top X_i - v_0) \ge 0\right), ~ \sum_{i=1}^n I\left(r_i(\bm\theta) (\vv_1^\top X_i - v_0) \le 0\right)\right\}\\
   && \quad=\inf_{v_0\in \mathbb{R}^1, \vv_1 \in \mathcal{S}^{d-1}} \frac{1}{n} \left\{\sum_{i=1}^n I\left(r_i(\bm\theta) (\vv_1^\top X_i - v_0) \le 0\right)\right\}\\
   && \quad=\inf_{v_0\in \mathbb{R}^1, \vv_1 \in \mathcal{S}^{d-1}} \frac{1}{n} \left\{\sum_{i=1}^n I\left(r_i(\bm\theta) {-v_0 \choose \vv_1}^\top {1 \choose X_i} \le 0\right)\right\}\\
   && \quad=\inf_{v_0\in \mathbb{R}^1, \vv_1 \in \mathcal{S}^{d-1}} \frac{1}{n} \left\{\sum_{i=1}^n I\left(r_i(\bm\theta) {-v_0 \choose \vv_1}^\top {1 \choose X_i} \Big/ \left\|{-v_0 \choose \vv_1}\right\| \le 0\right)\right\}\\
   && \quad=\inf_{\u \in \mathcal{S}^{d}_*} \frac{1}{n} \left\{\sum_{i=1}^n I\left(r_i(\bm\theta) (\u^\top W_i) \le 0\right)\right\}.
\end{eqnarray*}
Here $\mathcal{S}^{d}_* = \mathcal{S}^{d} \setminus \{(1, 0, 0, \cdots, 0)^\top, (-1, 0, 0, \cdots, 0)^\top\}$ because we have $\vv_1 \neq \0$ for any $\vv_1 \in \mathcal{S}^{d-1}$.

Next, for any given $n$ and $\bm\theta$, by \cite{LZ2014a}, we claim that the minimum value of $G(\u) = \frac{1}{n} \left\{\sum_{i=1}^n I\left(r_i(\bm\theta) (\u^\top W_i) \le 0\right)\right\}$ for any $\u \in \mathcal{S}^{d}$ occurs at some direction vectors contained in the inner points of an open set on $\mathcal{S}^{d}$. Hence,
\begin{eqnarray*}
  \inf_{\u \in \mathcal{S}^{d}_*} G(\u) \leq g(\u^*)
\end{eqnarray*}
for $\u^* \in \{(1, 0, 0, \cdots, 0)^\top, (-1, 0, 0, \cdots, 0)^\top\}$. That is,
\begin{eqnarray*}
  \HRD(\bm\theta, \bar{P}_n) &=& \inf_{\u \in \mathcal{S}^{d}} \frac{1}{n} \left\{\sum_{i=1}^n I\left(r_i(\bm\theta) (\u^\top W_i) \le 0\right)\right\}\\
  &=& \inf_{\u \in \mathcal{S}^{d}} \frac{1}{n} \left\{\sum_{i=1}^n I\left(r_i(\bm\theta) (\u^\top W_i) \le \u^\top \0\right)\right\}\\
  &=& \HD(\0, \bar{P}_{\theta, n}).
\end{eqnarray*}
This completes the proof.
\end{proof}

\begin{proof}[Proof of Theorem \ref{ThHRD}]
  The proof of Properties Q1-Q2 is trivial. We only check Q4. For any $\bm\theta \neq \bm\theta_0$, let $\u_{\theta} = (\bm\theta_0 - \bm\theta) / \|\bm\theta_0 - \bm\theta\|$. Observe that
  \begin{eqnarray*}
    \u_{\theta}^\top (Y - W^\top \bm\theta) W &=& \u_{\theta}^\top \varepsilon W + \u_{\theta}^\top WW^\top (\bm\theta_0 - \bm\theta) \\
    &=& \u_{\theta}^\top \varepsilon W + (\u_{\theta}^\top W)^2 \times \|\bm\theta_0 - \bm\theta\| \rightarrow +\infty
  \end{eqnarray*}
  in probability 1 as $\|\bm\theta\| \rightarrow \infty$. Hence,
  \begin{eqnarray*}
    \HRD(\bm\theta, \bar{P}) \le P\left(\u_{\theta}^\top(Y - W^\top \bm\theta)W \le 0\right) \rightarrow 0,
  \end{eqnarray*}
  as $\|\bm\theta\| \rightarrow \infty$. This completes the proof.
\end{proof}

\begin{proof}[Proof for Remark \ref{remark2}]
For the true parameter $\bm\theta_0$, since $\varepsilon = Y - W^\top \bm\theta_0$, we have
\begin{eqnarray*}
  P(\u^\top (Y - W^\top \bm\theta_0) W \le 0) &=& E_X(P(\u^\top (Y - W^\top \bm\theta_0) W \le 0|X))\\
  & = & E_X\left(\frac{1}{2} I(\u^\top W < 0) + \frac{1}{2} I(\u^\top W > 0) + I(\u^\top W = 0)\right)\\
   &=& \frac{1}{2},
\end{eqnarray*}
for any $\u \in \mathcal{S}^d$, where $E_X$ denotes the expectation taken with respect to $X$. This proofs that $Z = (X^\top, Y)^\top$ is regression halfspace symmetrically distributed about $\bm\theta_0$.

Next, for any $\bm\theta \neq \bm\theta_0$, observe that
\begin{eqnarray*}
  &&P(\vv_{\theta}^\top (Y - W^\top \bm\theta) W \le 0)\\
  && = P(\varepsilon \vv_{\theta}^\top W \le \vv_{\theta}^\top WW^\top (\bm\theta - \bm\theta_0))\\
  && = E_X\left(P(\varepsilon \ge -\|\bm\theta - \bm\theta_0\| \cdot \vv_{\theta}^\top W|X) I(\vv_{\theta}^\top W < 0) + P(\varepsilon \le -\|\bm\theta - \bm\theta_0\| \cdot \vv_{\theta}^\top W|X) I(\vv_{\theta}^\top W > 0)\right)\\
  && < \frac{1}{2},
\end{eqnarray*}
where $\vv_{\theta} = (\bm\theta_0 - \bm\theta) / \|\bm\theta_0 - \bm\theta\|$, because $X$ is continuous. Hence, the regression halfspace symmetrical center $\bm\theta_0$ is unique.
\end{proof}

\begin{proof}[Proof of Theorems \ref{th:SRD}-\ref{th:ZRD}]
  Using similar techniques to Theorem \ref{ThHRD}, the check of Q1, Q2 and Q4 for Theorems \ref{th:SRD}-\ref{th:ZRD} follows the same fashion. We only check Q3 for Theorem \ref{th:RRD}.

  Observe that, for any $\bm\theta \neq \bm\theta_0$,
  \begin{eqnarray*}
    \tilde{O}(\bm\theta, \bar{P}) & = & \sup_{\u \in\mathcal{S}^d} \frac{E\left(\u^\top (Y - W^\top\bm\theta)W\right)}{\sqrt{\text{Var}(\u^\top YW)}}\\
    & = & \sup_{\u \in\mathcal{S}^d} \frac{\u^\top \bar{\Sigma} (\bm\theta_0 - \bm\theta)}{\sqrt{\u^\top \tilde{\Sigma} \u}}\\
    & = & \sup_{\vv \in\mathcal{S}^d} \vv^\top \tilde{\Sigma}^{-1/2} \bar{\Sigma} (\bm\theta_0 - \bm\theta)\\
    & = & \|\tilde{\Sigma}^{-1/2} \bar{\Sigma} (\bm\theta_0 - \bm\theta)\|,
  \end{eqnarray*}
  where $\tilde{\Sigma} = E((YW - E(YW))(YW - E(YW))^\top)$. Hence, $\tilde{O}(\bm\theta_0 + \lambda(\bm\theta_0 - \bm\theta), \bar{P}) = (1 - \lambda) \|\tilde{\Sigma}^{-1/2} \bar{\Sigma} (\bm\theta_0 - \bm\theta)\| \leq \tilde{O}(\bm\theta, \bar{P})$, which implies $\PRD^r(\bm\theta_0 + \lambda(\bm\theta_0 - \bm\theta), \bar{P}) \ge \PRD^r(\bm\theta, \bar{P})$ for any $\lambda \in [0, 1]$. This completes the proof.
\end{proof}

\end{document}